\documentclass[11pt,times]{article}
\usepackage[utf8]{inputenc}
\usepackage{amssymb,amsmath}
\usepackage{algorithm}
\usepackage{algpseudocode}
\usepackage{amsthm}
\usepackage{enumerate}
\usepackage[font={scriptsize,it}]{caption}
\usepackage{subcaption}
\allowdisplaybreaks
\usepackage[T1]{fontenc}

\usepackage{float}
\newfloat{algorithm}{t}{lop}

\usepackage{tikz}
\usetikzlibrary{decorations.pathmorphing}
\usetikzlibrary{shapes,fit}
\usetikzlibrary{backgrounds}

\usepackage{bm}
\usepackage{multirow}

\usepackage[font=scriptsize,skip=5pt]{caption}
\usepackage[inline]{enumitem}

\newtheorem{theorem}{Theorem}[section]
\newtheorem{corollary}[theorem]{Corollary}
\newtheorem{definition}[theorem]{Definition}
\newtheorem{lemma}[theorem]{Lemma}

\newtheorem*{theorem*}{Theorem}
\newtheorem*{corollary*}{Corollary}
\newtheorem*{lemma*}{Lemma}
\newtheorem*{observation*}{Observation}

\newcommand{\Xomit}[1]{}
\newcommand{\vone}{\vspace{.1in}}

\newcommand{\la}{\leftarrow}

\newcommand{\congest}{{\sc Congest}}
\newcommand{\Tau}{H}
\newcommand{\hide}[1]{}
\newcommand{\centers}{blocker set }

\newcommand{\boi}{\begin{enumerate}}
\newcommand{\eoi}{\end{enumerate}}
\newcommand{\bii}{\begin{itemize}}
\newcommand{\eii}{\end{itemize}}

\newcommand{\R}{Randomized }
\newcommand{\A}{Arbitrary }
\newcommand{\I}{Integer }
\newcommand{\D}{Deterministic }
\newcommand{\DI}{Directed \& Undirected }

\renewcommand\labelenumi{(\roman{enumi})}
\renewcommand\theenumi\labelenumi

\algtext*{EndWhile}
\algtext*{EndFor}
\algtext*{EndIf}

\advance \topmargin by -\headheight
\advance \topmargin by -\headsep
\evensidemargin \oddsidemargin
\begin{document}

\title{New and Simplified Distributed Algorithms for Weighted All Pairs Shortest Paths}
\author{Udit Agarwal $^{\star}$ and Vijaya Ramachandran\thanks{Dept. of Computer Science, University of Texas, Austin TX 78712. Email: {\tt udit@cs.utexas.edu, vlr@cs.utexas.edu}. This work was supported in part by NSF Grant CCF-1320675. }}

\maketitle
\begin{abstract}
We consider the problem of computing all pairs shortest paths (APSP) and
shortest paths for $k$ sources
in a weighted graph in the distributed {\sc Congest}  model. For graphs with non-negative integer
edge weights (including zero weights) we build on a recent pipelined algorithm~\cite{AR18}
to obtain a
$\tilde{O}(\lambda^{1/4}\cdot  n^{5/4} )$-round bound
 for graphs with  edge-weight at most $\lambda$, and
$\tilde{O}(n \cdot \bigtriangleup^{1/3})$-round bound
 for shortest path distances at most $\bigtriangleup$.
Additionally, we  simplify some of the procedures in the earlier APSP algorithms for non-negative edge weights in~\cite{HNS17,ARKP18}.
We also present results for computing $h$-hop shortest paths and shortest paths from $k$ given sources.
 
In other results, we present a randomized exact APSP algorithm for graphs with arbitrary edge weights that runs in
$\tilde{O}(n^{4/3})$ rounds w.h.p. in $n$, which improves the previous best $\tilde{O}(n^{3/2})$ bound, which is deterministic. We also present
an $\tilde{O}(n/\epsilon^2)$-round deterministic $(1+\epsilon)$ approximation algorithm for graphs with non-negative $poly(n)$ integer weights (including
zero edge-weights),
improving  results in~\cite{Nanongkai14,LP15} that hold only for positive integer weights.
\end{abstract}

\section{Introduction}\label{sec:intro}

 Designing distributed algorithms for various network and graph problems 
such as shortest paths~\cite{ARKP18,HNS17,LP13,PR18,KN18}
is
a extensively studied area of research.
The {\sc Congest} model (described in Sec~\ref{sec:congest})
is a  
widely-used
model for these algorithms,
see~\cite{ARKP18,Elkin17,HNS17,LP13}.
In this paper we consider  distributed algorithms for the 
 computing all pairs shortest paths (APSP) and  related problems in a graph  with non-negative edge weights
 in the {\sc Congest} model.

In sequential computation, shortest paths can be computed much faster in graphs with non-negative
edge-weights (including zero weights) using the classic Dijkstra's algorithm~\cite{Dijkstra59} than in
graphs with
negative edge weights. Additionally, negative edge-weights raise the possibility
of negative weight cycles in the graph, which usually do not
occur in practice, and hence are not modeled by real-world weighted graphs. Thus, in the distributed
setting, it is of importance to design fast shortest path algorithms that can handle non-negative
edge-weights, including edges of weight zero.

 The presence of zero weight edges creates challenges in the design of distributed algorithms  
as  observed in~\cite{HNS17}.
 (We review related work in Section~\ref{sec:related}.)
 One approach used for positive integer edge weights is to replace an edge of weight $d$ with $d$ 
 unweighted edges and then run an unweighted APSP algorithm such as~\cite{LP13,PR18} on this modified graph.
 This approach is used in approximate APSP algorithms~\cite{Nanongkai14, LP15}.
 However such an approach fails when zero weight edges may be present.   
 There are a few known algorithms that can handle zero weights, such as the $\tilde{O}(n^{5/4})$-round randomized
 APSP algorithm of Huang et al.~\cite{HNS17} (for polynomially bounded non-negative
 integer edge weights) and the
 $\tilde{O}(n^{3/2})$-round deterministic APSP algorithm of Agarwal et al.~\cite{ARKP18} (for graphs with arbitrary edge weights including zero weights). A deterministic pipelined algorithm for this problem that runs in at most $ 2 \cdot n \sqrt{\Delta} + 2n$ was
 recently given in~\cite{AR18}, where $\Delta$ is an upper bound on the shortest path length.

 \subsection{Our Results}

 We present several new results for computing APSP and related problems on  an $n$-node graph $G=(V,E)$ with non-negative edge 
 weights $w(e), e\in E$, including deterministic distributed sub-$n^{3/2}$-round algorithms for moderate
  weights (including zero weights)~\cite{AR18}. 
 All of our results hold for both directed and undirected graphs and we will assume
 w.l.o.g. that $G$ is directed. 
 
 Many of our results build on a recent deterministic distributed  pipelined algorithm we developed for APSP and $k$-SSP for 
graphs with non-negative integer weights (including zero weights)~\cite{AR18}. This algorithm computes the
$h$-hop
shortest path problem for $k$ sources
($(h,k)$-SSP), 
 with an additional constraint that the shortest paths have distance at most $\bigtriangleup$ in $G$,
 together with the corresponding shortest path trees, defined as follows.

\begin{definition}\label{def:h-sssp}
An \emph{$h$-hop shortest path from $u$ to $v$ in $G$}
is a path from $u$ to $v$ of minimum weight among all paths with at most $h$ edges (or {\it hops}).\\
  An  \emph{$h$-SSP tree for source $s$ and shortest path distance $\Delta$} is a tree rooted at $s$ that  contains  an $h$-hop shortest path
   from $s$ to every other vertex to which
  there exists an $h$-hop  path with weight at most $\bigtriangleup$ in $G$.
In the case of multiple $h$-hop shortest paths from $s$ to a vertex $v$, this tree contains the path with the smallest number of hops, breaking any further ties by choosing the predecessor vertex with smallest ID.
\end{definition}

The pipelined algorithm
 achieves the bounds in the following theorem.

 \begin{theorem}\label{thm:alg2}
\cite{AR18}  Let $G=(V,E)$ be a directed or undirected edge-weighted graph, where all edge weights are 
 non-negative integers (with zero-weight edges allowed).
 The following 
 deterministic bounds can be obtained  in the \congest{} model for shortest path distances at most 
 $\bigtriangleup$.\\
 (i) \emph{$(h,k)$-SSP} in $2 \sqrt{\bigtriangleup kh} + k + h $ rounds.\\
 (ii) \emph{APSP} in  $2n \sqrt{\bigtriangleup} + 2n$ rounds.\\
 (iii) \emph{$k$-SSP} in $2 \sqrt{\bigtriangleup kn} + n + k$ rounds.
 \end{theorem}

The new results we present in this paper are the following.
 
 \vone
 \noindent
 {\bf 1. Faster Deterministic APSP for Non-negative, Moderate Integer Weights.}
 We improve on the 
 bounds given in $(ii)$ and $(iii)$ of Theorem~\ref{thm:alg2}
by combining the pipelined algorithm in~\cite{AR18} with a modified version of the 
 APSP algorithm in~\cite{ARKP18} to obtain our improved Algorithm~\ref{algkSSP}, with the bounds
 stated in the following Theorems~\ref{thm:algkSSPEdgeBound} and \ref{thm:algkSSP}. 
 To obtain these improved 
bounds we also present an improved deterministic distributed algorithm to find
 a `blocker set'~\cite{ARKP18}.
 
 \vspace{-.03in}

\begin{theorem} \label{thm:algkSSPEdgeBound}
Let $G=(V,E)$ be a directed or undirected edge-weighted graph, 
where all edge weights are 
 non-negative integers bounded by $\lambda$ (with zero-weight edges allowed).
 The following deterministic bounds can be obtained in the \congest{} model.\\
 (i) APSP  in $O(\lambda^{1/4}\cdot  n^{5/4} \log^{1/2} n)$ rounds.\\
 (ii) $k$-SSP in $O(\lambda^{1/4}\cdot  nk^{1/4} \log^{1/2} n)$ rounds.
 \end{theorem}

\begin{theorem} \label{thm:algkSSP}
Let $G=(V,E)$ be a directed or undirected edge-weighted graph, where all edge weights are 
non-negative integers (with zero edge-weights allowed),
 and the shortest path distances are bounded by $\bigtriangleup$.
 The following 
 deterministic bounds can be obtained in the \congest{} model.\\
 (i) APSP  in $O(n (\bigtriangleup \log^2 n)^{1/3})$ rounds.\\
 (ii) $k$-SSP in $O((\bigtriangleup kn^2 \log^2 n)^{1/3})$ rounds.
 \end{theorem}

Our results in Theorem~\ref{thm:algkSSPEdgeBound} and \ref{thm:algkSSP}  improve on the 
$\tilde{O}(n^{3/2})$ deterministic APSP bound of Agarwal et al.~\cite{ARKP18} for significant ranges of values 
for both $\lambda$ and $\Delta$, as stated below.

   \begin{corollary} \label{cor:lambda}
Let $G=(V,E)$ be a directed or undirected edge-weighted graph with non-negative edge weights 
(and zero-weight edges allowed). 
 The following deterministic bounds hold for 
 the \congest{} model for $1\geq \epsilon\geq 0$. \\
 (i) If the edge weights are bounded by $\lambda = n^{1-\epsilon}$, then APSP can be computed in $O(n^{3/2 - \epsilon/4}\log^{1/2} n)$ rounds.\\
 (ii) For shortest path distances bounded by $\Delta = n^{3/2-\epsilon}$, APSP can be computed in $O(n^{3/2 - \epsilon/3} \log^{2/3} n)$ rounds.
 \end{corollary}

 The corresponding bounds for the weighted $k$-SSP problem are: $O(n^{5/4-\epsilon/4}k^{1/4}\log^{1/2} n)$ (when $\lambda = n^{1-\epsilon}$) and $O(n^{7/6 - \epsilon/3}k^{1/3}\log^{2/3} n)$ (when $\Delta = n^{3/2-\epsilon}$).
 Note that the result in $(i)$ is independent of the value of $\Delta$ (depends only on $\lambda$)
 and the result in $(ii)$ is independent of the value of $\lambda$ (depends only on $\Delta$).

 \vone
 \noindent
 {\bf 2. Simplifications to Earlier Algorithms.} Our techniques give simpler methods for some of procedures in the two previous
 distributed  weighted APSP algorithms that handle zero weight edges. In Section~\ref{sec:sssp} we present simple deterministic algorithms that
 match the congest and dilation bounds in~\cite{HNS17} for two of the three procedures used there: the {\it short-range} and {\it short-range-extension} algorithms. Our simplified algorithms are both obtained using a streamlined single-source version of the pipelined APSP algorithm in~\cite{AR18}.
 
 A key contribution in the deterministic APSP algorithm
  in~\cite{ARKP18} is a fast  deterministic distributed algorithm for computing a {\it blocker set}. The performance of the blocker set algorithm
  in~\cite{ARKP18} does not suffice for our faster APSP algorithms 
  (Theorems~\ref{thm:algkSSPEdgeBound} and \ref{thm:algkSSP}).
 In Section~\ref{sec:kSSP} we present a faster blocker set algorithm, which is also
 a simplification of the blocker set algorithm  in~\cite{ARKP18}. The improved bound that we obtain here for computing a blocker set will
 not improve the overall bound  in~\cite{ARKP18},  but our method could be used there to achieve the same bound with a more streamlined algorithm.

\vone
\noindent
{\bf 3. Faster (Randomized) APSP for Arbitrary Edge-Weights.} For exact APSP in directed graphs with arbitrary edge-weights the only
prior nontrivial result known is the $\tilde{O}(n^{3/2})$-round deterministic algorithm in~\cite{ARKP18}. We present an algorithm with the following
improved randomized bound in Section~\ref{sec:randAPSP}.

 \begin{theorem}	\label{thm:arbitrary}
 Let $G = (V,E)$ be a directed or undirected edge-weighted graph with arbitrary edge weights. 
 Then, we can compute  weighted APSP in $G$  in the \congest{} model
  in $\tilde{O}(n^{4/3})$ rounds, w.h.p. in $n$.
 \end{theorem}

  The corresponding bound for $k$-SSP is $\tilde{O}(n + n^{2/3}k^{2/3})$.

 \vone
 \noindent
 {\bf 4. Approximate APSP for Non-negative Edge Weights}.
 In Section~\ref{sec:approx} we present an algorithm that matches
  the earlier bound for computing approximate APSP in graphs with {\it positive} integer edge weights~\cite{Nanongkai14,LP15}
  by obtaining the same bound for non-negative edge weights.
  
    \vspace{-.03in}

\begin{theorem} \label{thm:algApprox}
 Let $G=(V,E)$ be a directed or undirected edge-weighted graph, where all edge weights are non-negative  integers  polynomially bounded in $n$, and where zero-weight edges are allowed. 
 Then,  for any $\epsilon>0$ we can compute $(1+\epsilon)$-approximate APSP in 
 $O((n/\epsilon^2) \cdot \log n)$ rounds deterministically  in the \congest{} model.
 \end{theorem}

\vone
{\bf Roadmap.} The rest of the paper is organized as follows. In Sections~\ref{sec:congest} and \ref{sec:related} we review the {\sc Congest} model and discuss related work.
In Section~\ref{sec:kSSP} we present our faster APSP and $k$-SSP deterministic distributed algorithms, including our improved deterministic method to compute a blocker set.
Section~\ref{sec:sssp}  describes our simple algorithms for the short-range and short-range extension problems from Huang et al.~\cite{HNS17}. Section~\ref{sec:additional}
presents our results that give Theorems~\ref{thm:arbitrary} and \ref{thm:algApprox}, and we conclude with Section~\ref{sec:conclusion}.

\begin{table*}
\scriptsize
\centering
\caption{Table comparing our new results for non-negative edge-weighted graphs (including zero edge weights) with previous known results. Here $\lambda$ is the maximum edge weight and $\Delta$ is the maximum weight of a shortest path in $G$.} \label{table:exact}
\def\arraystretch{1.5}
{\setlength{\tabcolsep}{.3em}
\begin{tabular}{| c | c |l| c | c | c |}
\hline
\multicolumn{6}{|c|}{\sc \textbf{Problem:}  Exact Weighted APSP} \\
\hline
Author & Arbitrary/  & handle zero & Randomized/ & Undirected/ & Round  \vspace{-.05in}\\
 & Integer weights  &  weights & Deterministic & (Directed \& Undirected) & Complexity \\
 \hline
Huang et al.~\cite{HNS17} & Integer & Yes & \R & Directed \& Undirected & $\tilde{O}(n^{5/4})$ \\
\hline
Elkin~\cite{Elkin17} & Arbitrary & Yes & \R & Undirected & $\tilde{O}(n^{5/3})$ \\
\hline
Agarwal et al.~\cite{ARKP18} & \A & Yes & \D & Directed \& Undirected & $\tilde{O}(n^{3/2})$ \\
\hline
 \multirow{3}{*}{\textbf{This paper}} & \multirow{2}{*}{\textbf{\I}}  & \multirow{2}{*}{\textbf{Yes}} & \multirow{2}{*}{\textbf{\D}} & \multirow{2}{*}{\textbf{Directed \& Undirected}} & $\mathbf{\tilde{O}(n^{3/2 - \epsilon/4})}$ \textbf{(when }$\mathbf{\lambda \leq n^{1-\epsilon}}$)\\
 &  & & &  & $\mathbf{\tilde{O}(n^{3/2 - \epsilon/3})}$ \textbf{(when} $\mathbf{\Delta \leq n^{3/2 - \epsilon}}$)\\\cline{2-6}
 & \textbf{\A} & \textbf{Yes} & \textbf{\R} & \textbf{\DI} & $\mathbf{\tilde{O}(n^{4/3})}$ \\
\hline
\hline
\multicolumn{6}{|c|}{\sc \textbf{Problem:} $(1+\epsilon)$-Approximation Weighted APSP} \\
\hline
 Nanongkai~\cite{Nanongkai14} & Integer & No & \R & \DI & $\tilde{O}(n/\epsilon^2)$ \\
\hline
Lenzen \&  & \I & No & \D & \DI & $\tilde{O}(n/\epsilon^2)$ \\
Patt-Shamir~\cite{LP15} & & & & & \\
\hline
\textbf{This paper} & \textbf{\I} & \textbf{Yes} & \textbf{\D} & \textbf{\DI} & $\mathbf{\tilde{O}(n/\epsilon^2)}$ \\
\hline
\end{tabular}}
\end{table*}

\subsection{Congest Model}	\label{sec:congest}

In the {\sc Congest} model,
there are $n$ independent processors 
interconnected in a network by bounded-bandwidth links.
We refer to these processors as nodes and
the links as edges.
This network is modeled by graph $G = (V,E)$
where $V$ refers to the set of processors and
$E$ refers to the set of links between the processors.
Here $|V| = n$ and $|E| = m$.

Each node is assigned a unique ID  
between 1 and $poly(n)$ and
has infinite computational power.
Each node has limited topological knowledge 
and only knows about its incident edges.
For the integer-weighted APSP problem we consider, 
 each edge 
has a non-negative
  integer weight (zero weights allowed) that can be represented with 
  $B= O(\log n)$ bits.
Also if the edges are directed, 
the corresponding communication channels are bidirectional
and hence the communication network can be represented 
by the underlying undirected graph $U_G$ of $G$
(this is also the assumption used in~\cite{HNS17,GL18,ARKP18}). 
The pipelined algorithm in~\cite{AR18} does not need this
feature, and uses only the directed edges in the graph for communication. 

The computation proceeds in rounds. In each round each processor can send an $O(\log n)$-bit message
along edges incident to it, and it receives the messages sent to it in the previous
round. 
(If the graph has arbitrary edge weights, a node can send a constant number of distance values and
node IDs along each edge in a message.)
The model allows a node to send different message along different edges though we do not
need this feature in our algorithm.
The performance of an algorithm in the \congest{} model is measured by its
round complexity, which is the worst-case
number of rounds of distributed communication.
\vspace{-.03in}

\subsection{Related Work}\label{sec:related}

\noindent
{\bf Weighted APSP.}
The current best bound for the weighted APSP problem is due to the randomized algorithm of Huang et al.~\cite{HNS17}
that runs in $\tilde{O}(n^{5/4})$ rounds.
This algorithm works for graphs with polynomially bounded integer edge weights (including
zero-weight edges), and the result holds with w.h.p. in $n$.
For graphs with arbitrary edge weights, the recent result of Agarwal et al.~\cite{ARKP18} gives a deterministic APSP
algorithm that runs in $\tilde{O}(n^{3/2})$ rounds.
This is the current best bound (both deterministic and randomized) for graphs with arbitrary edge weights as well as
 the best deterministic bound for graphs with integer edge weights.
 
 In this paper we present an algorithm for non-negative integer edge-weights (including zero-weighted
  edges) that runs in $\tilde{O}(n \bigtriangleup^{1/3})$ rounds 
where the shortest path distances are at most $\bigtriangleup$
and in $\tilde{O}(n^{5/4}\lambda^{1/4})$ rounds when the edge weights are bounded by $\lambda$.
This result improves on the $\tilde{O}(n^{3/2})$ deterministic APSP bound of Agarwal et al.~\cite{ARKP18} 
when either edge weights are at most $n^{1-\epsilon}$ or
 shortest path distances are at most $n^{3/2 - \epsilon}$, for any $\epsilon > 0$.
 
 We also give an improved randomized algorithm for APSP in graphs with arbitrary edge weights that
 runs in $\tilde{O}(n^{4/3})$ rounds, w.h.p. in $n$.

\vspace{.03in}

\noindent
{\bf  Weighted $k$-SSP.}
 The current best bound for the weighted $k$-SSP problem is due to the Huang et al's~\cite{HNS17} randomized
 algorithm that runs in $\tilde{O}(n^{3/4}\cdot k^{1/2} + n)$ rounds.
 This algorithm is also randomized and only works for graphs with integer edge weights.
 The recent deterministic APSP algorithm in~\cite{ARKP18} can be shown to give an
$O(n \cdot \sqrt{k \log n})$ round deterministic algorithm for $k$-SSP.
In this paper, we present a deterministic algorithm for positive including zero integer edge-weighted graphs
that runs in $\tilde{O}((\bigtriangleup \cdot n^2 \cdot k)^{1/3})$ rounds where the shortest path distances are at
 most $\bigtriangleup$ and in $\tilde{O}((\lambda k)^{1/4}n)$ rounds when the edge weights are bounded by
 $\lambda$.

\vspace{.03in}

\noindent
{\bf $(1+\epsilon)$-Approximation Algorithms.}
For graphs with positive integer edge weights, deterministic 
$\tilde{O}(n/\epsilon^2)$-round algorithms for
a $(1+\epsilon)$-approximation to APSP are known~\cite{Nanongkai14, LP15}.
But these algorithms do not handle zero weight edges.
In this paper we present a deterministic algorithm that handles zero-weight edges and
matches the $\tilde{O}(n/\epsilon^2)$-round bound for approximate APSP known before
for positive edge weights.

\section{Faster $k$-SSP Algorithm Using a Blocker Set}	\label{sec:kSSP}

In this section we  give faster deterministic  APSP and $k$-SSP algorithms than the $\tilde{O}(n^{3/2})$ bound in~\cite{ARKP18}
 for moderate non-negative edge weights (including zero weights).  The overall Algorithm~\ref{algkSSP}
has the same structure as the 
deterministic
${O}(n^{3/2} \cdot \sqrt{\log n}))$ round weighted APSP algorithm  in~\cite{ARKP18} but we use
a variant of the pipelined APSP algorithm in~\cite{AR18}  in place of Bellman-Ford, and we also present
new methods within two of the steps.

We first define the following notion of an  {\it $h$-hop Consistent SSSP (CSSSP)} collection. This notion is
implicit in~\cite{ARKP18} but is not explicitly defined there.

\vspace{-.05in}
 \begin{definition}[{\bf CSSSP}]	\label{def:CSSSP}
  Let $H$ be a collection of rooted $h$-hop trees in a graph $G=(V,E)$. Then $H$ is an \emph{$h$-hop CSSSP collection} (or simply an \emph{$h$-hop CSSSP})
  if for every $u, v \in V$ the path from $u$ to $v$ is 
  the same in each of the trees in $H$ (in which such a path exists), and is the $h$-hop shortest path from $u$ to $v$ in the $h$-hop tree $T_u$ rooted at $u$. 
  Further, each $T_u$ contains every vertex $v$ that has a shortest path from $u$ in $G$ with at
  most $h$ hops.
  \end{definition}
  
  \vspace{-.05in}

  In our improved Algorithm~\ref{algkSSP}, Steps~\ref{algkSSP:computeSSSP}-\ref{algkSSP:local} 
   are unchanged from the algorithm in~\cite{ARKP18}.
However 
we give an alternate method for 
Step~\ref{algkSSP:computeh-hop} to compute $h$-hop CSSSP (see Section~\ref{sec:csssp})
 since the method in~\cite{ARKP18} takes $\Theta (n \cdot h)$ rounds, which is
too large for our purposes. Our new method is very simple and using the pipelined algorithm in~\cite{AR18} it runs in
$O( \sqrt{\bigtriangleup hk})$ rounds.
(An implementation using Bellman-Ford~\cite{Bellman58} would give an $O(n \cdot h)$-round bound, which could be used in~\cite{ARKP18} to simplify that blocker set algorithm.)

  Step~\ref{algkSSP:computeBlocker} computes a
{\it blocker set}, defined as follows.

 \vspace{-.05in}

  \begin{definition}[{\bf Blocker Set}~\cite{King99,ARKP18}]
  Let $\Tau$ be a collection of rooted $h$-hop trees in a graph $G=(V,E)$. A set $Q\subseteq V$ is a 
\emph{\centers} for $\Tau$ if every root to leaf path 
 of length $h$ 
  in every tree in $\Tau$ contains a vertex in $Q$.
  Each vertex in $Q$ is called a \emph{blocker vertex} for $\Tau$.
  \end{definition}
  
  \vspace{-.05in}
 
For  Step~\ref{algkSSP:computeBlocker} we use the overall blocker set
algorithm from~\cite{ARKP18}, which runs in $O(n \cdot h + (n^2 \log n)/h)$ rounds and
computes a blocker set of size $q=O((n \log n)/h)$ for the $h$-hop trees constructed in Step~\ref{algkSSP:computeh-hop}
of algorithm~\ref{algkSSP}.
But this gives only an $\tilde{O}(n^{3/2})$ bound for
 Step~\ref{algkSSP:computeBlocker} (by setting $h= \tilde{O}(\sqrt n)$), so it
  will not help us to improve the bound on the number of rounds for APSP.
 Instead,  we modify and improve a key step where that  earlier blocker set algorithm 
has a $\Theta (n \cdot h)$ round preprocessing step.
(Our improved  method here
will not help to improve the bound in~\cite{ARKP18} but does help
to obtain a better bound here in conjunction with the pipelined algorithm.)
We give the details of  our method for  Step~\ref{algkSSP:computeBlocker}  in Section~\ref{sec:blocker}.  

 \vspace{-.05in}

\begin{algorithm}[H]
\scriptsize
\caption{\scriptsize Overall $k$-SSP algorithm (adapted from~\cite{ARKP18})}
Input: set of sources $S$, number of hops $h$
\begin{algorithmic}[1]
\State Compute $h$-hop CSSSP rooted at each source $x \in S$ (described in Section~\ref{sec:csssp}).
\label{algkSSP:computeh-hop}
\State Compute a blocker set $Q$ of size $\Theta(\frac{n\log n}{h})$ for the $h$-hop CSSSP
 computed in Step~\ref{algkSSP:computeh-hop}
(described in Section~\ref{sec:centers}).\label{algkSSP:computeBlocker}
\State {\bf for each $c \in Q$ in sequence:} compute SSSP tree rooted at $c$. \label{algkSSP:computeSSSP}
\State {\bf for each $c \in Q$ in sequence:} broadcast $ID(c)$ and the shortest path 
distance values $\delta_h (x,c)$ for each $x \in S$.		\label{algkSSP:broadcast}
\State {\bf Local Step at node $v \in V$:} for each $x \in S$  compute  the shortest path distance $\delta (x,v)$ using the received values.	\label{algkSSP:local}
\end{algorithmic}  \label{algkSSP}
\end{algorithm}
\vspace{-.15in}

\begin{lemma}	\label{lemma:kSSP}
Algorithm~\ref{algkSSP} computes $k$-SSP in $O(\frac{n^2\log n}{h} + \sqrt{\bigtriangleup hk})$ rounds.
 \end{lemma}
 \vspace{-.2in}
 \begin{proof}
 The correctness of Algorithm~\ref{algkSSP} is established in~\cite{ARKP18}.
Step~\ref{algkSSP:computeh-hop} runs in $O( \sqrt{\bigtriangleup hk})$ rounds by 
Lemma~\ref{lemma:hhopCSSSP} in Section~\ref{sec:csssp}.
 In Section~\ref{sec:blocker}  we will give an
$O(n\cdot q + \sqrt{\bigtriangleup hk})$
rounds algorithm to find a blocker set of size $q= O(\frac{n \log n}{h})$.
Simple $O(n \cdot q)$  round algorithms for Steps~\ref{algkSSP:computeSSSP} and \ref{algkSSP:broadcast} are given in~\cite{ARKP18}.
 Step~\ref{algkSSP:local} has no communication. Hence the overall bound for Algorithm~\ref{algkSSP} is 
 $O(n \cdot q + \sqrt{\bigtriangleup hk})$ rounds.
 Since  $q= O(\frac{n \log n}{h})$ this gives the desired bound.
 \end{proof}
 
\vspace{-.15in}

\begin{proof}[Proofs of Theorem~\ref{thm:algkSSPEdgeBound} and \ref{thm:algkSSP}:]
Using 
$h =  \frac{n^{4/3}\cdot \log^{2/3} n}{(2k\cdot \bigtriangleup)^{1/3}}$ in Lemma~\ref{lemma:kSSP}
we obtain the bounds in
Theorem~\ref{thm:algkSSP}.

If edge weights are bounded by $\lambda$, the weight of any $h$-hop path is at most $h\lambda$.
Hence by Lemma~\ref{lemma:kSSP}, the $k$-SSP algorithm (Algorithm~\ref{algkSSP}) runs in 
$O(\frac{n^2\log n}{h} + h\sqrt{\lambda k})$ rounds.
Setting $h = n\log^{1/2} n/(\lambda^{1/4} k^{1/4})$ we obtain the bounds stated in Theorem~\ref{thm:algkSSPEdgeBound}.
\end{proof}

\vspace{-.05in}

\subsection{Computing Consistent h-hop trees}\label{sec:csssp}

In Section~\ref{sec:intro} we defined a natural notion of an $h$-hop SSSP tree rooted at a source $s$,
as a rooted tree which contains an $h$-hop shortest path from $s$ to every other vertex to
which there exists a path from $s$ with at most $h$ hops. We also defined tie-breaking rules for the case
when  multiple paths from $s$ to $v$ satisfy this definition: a path with the smallest number of hops
is chosen, with further ties broken by choosing the predecessor vertex with smallest ID.
Each $h$-hop tree constructed by the pipelined $(h,k)$-SSP algorithm~\cite{AR18}
 satisfies the definition of  an $h$-hop SSSP
tree (as constructed using the $Z.p$ pointers)
and these trees can also be constructed for each source using the Bellman-Ford algorithm~\cite{Bellman58}.

The definition of a CSSSP collection (Def.~\ref{def:CSSSP}) places additional
stringent conditions on the
structure of $h$-hop SSSP trees in the collection, and neither the pipelined algorithm in~\cite{AR18}  nor Bellman-Ford
is guaranteed to construct this collection.  At the same time, the trees in a CSSSP collection
may not satisfy the definition of $h$-hop SSSP  in Sec.~\ref{sec:intro} since we may not have
a path from vertex $u$ to vertex $v$ in a $h$-hop CSSSP tree even if there exists a path from
$u$ to $v$ with at most $h$ hops. See Fig.~\ref{fig:example}.

  \begin{figure*}
\scriptsize
	\centering
    	\begin{subfigure}{.2\textwidth}
    	\centering
            	\begin{tikzpicture}[scale=0.6, every node/.style={circle, draw, inner sep=0pt, minimum width=5pt, align=center}]
            	\node[circle,draw,inner sep=0pt] (a)[label=below:$a$] at (-4,2) {}; 
            	\node[circle,draw,inner sep=0pt] (b)[label=above:$b$] at (-2,2) {}; 
            	\node[circle,draw,inner sep=0pt] (d)[label=below:$d$] at (-3,0) {}; 
            	\node[circle,draw,inner sep=0pt] (c)[label=below:$c$] at (-1,0) {}; 
            	\begin{scope}[->,every node/.style={ right}]
            	\draw[->] (a) -- (b) node[midway,label=above:$1$] {};
            	\draw[->] (b) -- (d) node[midway,label=left:$1$] {};
            	\draw[->] (b) -- (c) node[midway,label=right:$8$] {};
            	\draw[->] (d) -- (c) node[midway,label=below:$1$] {};
            	\end{scope}
            	\end{tikzpicture}
            	\caption*{(i) Example graph $G$.}
            	\label{fig:G}
		\end{subfigure}    
		~
		\begin{subfigure}{.35\textwidth}
		\centering
			\begin{tikzpicture}[scale=0.6, every node/.style={circle, draw, inner sep=0pt, minimum width=5pt, align=center}]
            	\node[circle,draw,inner sep=0pt] (a)[label=below:$a$] at (-4,2) {}; 
            	\node[circle,draw,inner sep=0pt] (b)[label=above:$b$] at (-2,2) {}; 
            	\node[circle,draw,inner sep=0pt] (d)[label=below:$d$] at (-3,0) {}; 
            	\node[circle,draw,inner sep=0pt] (c)[label=below:$c$] at (-1,0) {}; 
            	\node[circle,draw,inner sep=0pt] (b1)[label=above:$b$] at (2,2) {}; 
            	\node[circle,draw,inner sep=0pt] (d1)[label=below:$d$] at (1,0) {}; 
            	\node[circle,draw,inner sep=0pt] (c1)[label=below:$c$] at (3,0) {}; 
            	\begin{scope}[->,every node/.style={ right}]
            	\draw[->] (a) -- (b)  {};
            	\draw[->] (b) -- (d) {};
            	\draw[->] (b) -- (c)  {};
            	\draw[->] (b1) -- (d1) {};
            	\draw[->] (d1) -- (c1)  {};
            	\end{scope}
            	\end{tikzpicture}
            	\caption*{(ii) $2$-hop SSSP for source nodes $a$ (on left) and $b$ (on right).}
            	\label{fig:2hopSSSP}
   	  \end{subfigure}
   	  ~
   	  \begin{subfigure}{.35\textwidth}
		\centering
			\begin{tikzpicture}[scale=0.6, every node/.style={circle, draw, inner sep=0pt, minimum width=5pt, align=center}]
            	\node[circle,draw,inner sep=0pt] (a)[label=below:$a$] at (-4,2) {}; 
            	\node[circle,draw,inner sep=0pt] (b)[label=above:$b$] at (-2,2) {}; 
            	\node[circle,draw,inner sep=0pt] (d)[label=below:$d$] at (-3,0) {}; 
            	\node[circle,draw,inner sep=0pt] (b1)[label=above:$b$] at (2,2) {}; 
            	\node[circle,draw,inner sep=0pt] (d1)[label=below:$d$] at (1,0) {}; 
            	\node[circle,draw,inner sep=0pt] (c1)[label=below:$c$] at (3,0) {}; 
            	\begin{scope}[->,every node/.style={ right}]
            	\draw[->] (a) -- (b)  {};
            	\draw[->] (b) -- (d) {};
            	\draw[->] (b1) -- (d1) {};
            	\draw[->] (d1) -- (c1)  {};
            	\end{scope}
            	\end{tikzpicture}
            	\caption*{(iii) $2$-hop CSSSP for source set $\{a,b\}$.}
            	\label{fig:2hopSSSP}
   	  \end{subfigure}
    \caption{This figure illustrates an example graph $G$ where the collection of $2$-hop SSSPs is different from the 
$2$-hop CSSSP generated for the source set $S = \{a, b\}$.Observe that the edge $(b,c)$ is part of the $2$-hop SSSP rooted at $a$ (Fig. (ii)) but is not part of the $2$-hop CSSSP (Fig. (iii)) since there is a shorter path from $b$ to $c$ of hop-length $2$ (path $\langle b, d, c\rangle$ in $2$-SSP rooted at $b$).}
    \label{fig:example}
\end{figure*}
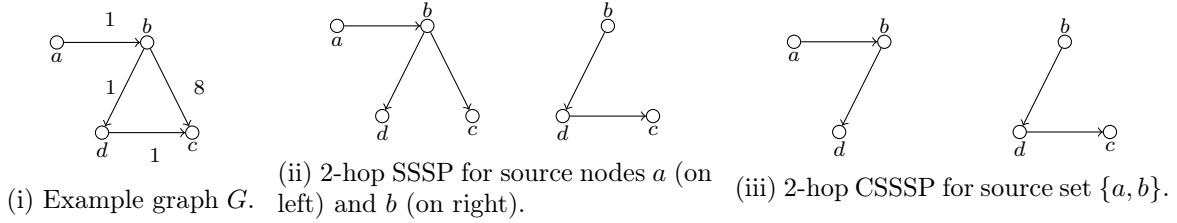

Our method to construct an $h$-hop CSSSP collection is very simple: We execute the pipelined algorithm in~\cite{AR18}
to construct $2h$-hop SSSP trees instead of $h$-hop SSSP trees. 
Our CSSSP collection will retain the initial $h$ hops of each of these $2h$-hop SSSP trees.
We now show that this simple construction results in an $h$-hop CSSSP collection.
Thus 
we are able to construct $h$-hop CSSSPs by incurring
just a constant factor overhead in the number of rounds over the bound for pipelined algorithm.

\vspace{-.05in}
\begin{lemma}	\label{lemma:CSSSP}
Let $\mathcal{A}$ be a distributed algorithm that computes $(h,k)$-SSP trees of shortest path distance at most $\Delta$ in an $n$-node graph in $f(h,k,n,\Delta)$ round.
Consider running Algorithm~$\mathcal{A}$ using the hop-length bound $2h$, and let
 $\mathcal{C}$ be the collection of $h$-hop trees formed by retaining the initial $h$ hops in each of these  $2h$-hop trees.
Then the collection $\mathcal{C}$ forms an $h$-hop CSSSP collection, and this collection can be computed in $f(2h,k,n,\Delta)$ rounds.
\end{lemma}

\vspace{-.2in}
\begin{proof}
If not, then there exist vertices $u, v$ and trees $T_x, T_y$ such that the paths from $u$ to $v$ in $T_x$ and
$T_y$ are different.
Let $\pi^x_{u,v}$ and $\pi^y_{u,v}$ be the corresponding paths in these trees.

There are three possible cases:
(1) when $wt(\pi^x_{u,v}) \neq wt(\pi^y_{u,v})$
(2) when paths $\pi^x_{u,v}$ and $\pi^y_{u,v}$ have same weight but different hop-lengths
(3) when both $\pi^x_{u,v}$ and $\pi^y_{u,v}$ have same weight and hop-length.

\textit{(1) $wt(\pi^x_{u,v}) \neq wt(\pi^y_{u,v})$:} 
w.l.o.g. assume that $wt(\pi^x_{u,v}) < wt(\pi^y_{u,v})$.
Now if we replace $\pi^y_{u,v}$ in $T_y$ with $\pi^x_{u,v}$, we get a path of smaller weight from $y$ to $v$ of
hop-length at most $2h$
and weight at most $\Delta$.
But this violates the definition of $h$-SSP (Definition~\ref{def:h-sssp}) since $T_y$ is a $2h$-SSP and hence it should contain a minimum
weight path from $y$ to $v$ of hop-length at most $2h$, and not the path $\pi^y_{y,v}$, resulting in a contradiction.

\textit{(2) paths $\pi^x_{u,v}$ and $\pi^y_{u,v}$ have same weight but different hop-lengths.}
w.l.o.g. assume that path $\pi^x_{u,v}$ has smaller hop-length than $\pi^y_{u,v}$.
Then $\kappa (\pi^x_{u,v}) < \kappa (\pi^y_{u,v})$ and hence 
$\kappa (\pi^y_{y,u} \circ \pi^x_{u,v}) < \kappa (\pi^y_{y,u} \circ  \pi^y_{u,v})$.
This again violates the $h$-SSP definition (Definition~\ref{def:h-sssp}) since $T_y$ is a $2h$-SSP and hence it should contain the path
$\pi^y_{y,u} \circ \pi^x_{u,v}$ from $y$ to $v$ as it has smaller number of hops than the path $\pi^y_{y,v}$.

\textit{(3) both $\pi^x_{u,v}$ and $\pi^y_{u,v}$ have same weight and hop-length.}
Let $(a,v)$ be the last edge on the path $\pi^x_{u,v}$ and let $(b,v)$ be the last edge on the path $\pi^y_{u,v}$.
w.l.o.g. assume that $ID(a) < ID(b)$.
Then again since $T_y$ is a $2h$-SSP, by $h$-SSP definition (Definition~\ref{def:h-sssp}) $T_y$ must contain the path 
$\pi^y_{y,u} \circ \pi^x_{u,v}$ from $y$ to $v$ since its predecessor vertex has smaller ID than the predecessor vertex of path $\pi^y_{y,v}$.
\end{proof}

\vspace{-0.1in}
\begin{lemma}	\label{lemma:hhopCSSSP}
An $h$-hop CSSSP collection can be computed in $O(\sqrt{\Delta hk})$ rounds using 
the pipelined APSP algorithm in~\cite{AR18}
 and in $O(nh)$ rounds using Bellman-Ford~\cite{Bellman58}.
\end{lemma}

\vspace{-.05in}
We now show two useful properties of an $h$-hop CSSSP collection that we will use in our
blocker set algorithm in the next section.
(Lemma~\ref{lemma:cInTree} is also implicitly established in~\cite{ARKP18}).

\vspace{-.05in}

\begin{lemma}\label{lemma:cOutTree}
Let $c$ be a vertex in $G$ and
let $T$ be the union of the edges in the collection of subtrees rooted at $c$ in the trees in 
a $h$-hop CSSSP collection $\mathcal{C}$.
Then $T$ forms an out-tree rooted at $c$.
\end{lemma}

\vspace{-.15in}
\begin{proof}
If not, there exist nodes $u$ and $v$ and trees $T_x$ and $T_y$ such that the path from $c$ to 
$u$ in $T_x$ and path from $c$ to $v$ in $T_y$ first diverge from 
each other after starting from $c$ and then coincide again at some vertex $z$. 
But since  $\mathcal{C}$ is an $h$-hop CSSSP collection,
 by Lemma~\ref{lemma:CSSSP} the path from $c$ to $z$ in the collection $\mathcal{C}$ is unique.
\end{proof}

\vspace{-.1in}

\begin{lemma}	\label{lemma:cInTree}
Let $c$ be a vertex in $G$ and let
$T$ be the collection of paths from source node $x \in S$ to  $c$ in the
trees in a $h$-hop CSSSP collection $\mathcal{C}$.
Then $T$ forms an in-tree rooted at $c$.
\end{lemma}

\vspace{-.1in}
\subsection{Computing a Blocker Set}\label{sec:blocker}\label{sec:centers}

Our overall blocker set algorithm runs in 
$O(\frac{n^2\log n}{h} + \sqrt{\bigtriangleup hk})$
 rounds. It differs from the 
blocker set algorithm in~\cite{ARKP18} by developing faster algorithms 
for two  steps that take $O(nh)$ rounds
in~\cite{ARKP18}. 

The first step in~\cite{ARKP18} that takes $O(nh)$ rounds is the step that computes the initial
`scores'  at all nodes for all $h$-hop trees in the CSSSP collection. 
The score of node $v$ in an $h$-hop tree is the number of
$v$'s descendants in that tree. Instead of the $O(nh)$ rounds, we can compute scores for all trees at all nodes
in $O(\sqrt{\bigtriangleup hk})$ rounds with  a timestamp technique 
given in~\cite{PR18} for
propagating values from descendants to ancestors in
the shortest path trees within the same bound as the APSP algorithm.

To explain the second $O(nh)$-round step in~\cite{ARKP18}, we first give 
 a brief recap of the blocker set algorithm in~\cite{ARKP18}.
This algorithm picks nodes to
be added to the blocker set greedily. The next node that is added to the blocker set is one that lies in the maximum number
of  paths in the $h$-hop trees that have not yet been covered by the already selected blocker nodes. To identify such a node,
the algorithm maintains at each node $v$ a count (or {\it score}) of the number of descendant leaves in each tree, since the sum of these
counts is precisely the number of root-to-leaf paths in which $v$ lies. Once all vertices have their overall score, the new blocker
node $c$ can be identified as one with the maximum score.  It now remains for each node $v$ to update its scores to reflect the
fact that paths through $c$ no longer exist in any of the trees. This update computation  is divided into two steps in~\cite{ARKP18}.
In both steps, the main challenge is for a given node to determine, in each tree $T_x$, whether it is an ancestor of $c$, a descendant of
$c$, or unrelated to $c$.

\vspace{.03in}

\noindent
1. {\it Updates at Ancestors.} For each $v$, in each tree $T_x$ where $v$ is an ancestor of $c$, $v$ needs to reduce its score for
$T_x$ by $c$'s score for $T_x$ since all of those descendant leaves have been eliminated. In~\cite{ARKP18} an $O(n)$-round
pipelined algorithm (using the
in-tree property in Lemma~\ref{lemma:cInTree})  is given for this update at all nodes in all trees, and this suffices for our purposes. 

\vspace{.03in}

\noindent
\noindent
2. {\it Updates at Descendants.} For each $v$, in each tree $T_x$ where $v$ is a descendant of $c$, $v$ needs to reduce its
score for $T_x$ to zero, since all descendant leaves are eliminated once $c$ is removed.
In~\cite{ARKP18} this computation is performed by an $O(nh)$-round precomputation in which 
each vertex identifies all of its ancestors in all of the $h$-hop trees and thereafter can readily identify
the trees in which it is a descendant of a newly chosen blocker node $c$ once $c$ broadcasts its
identity to all nodes. But this is too expensive for our purposes. 

\vspace{-.05in}

\begin{algorithm}[H]
\scriptsize
\caption{\scriptsize Pipelined Algorithm for updating scores at $v$ in trees $T_x$ in which $v$ is a descendant of newly chosen blocker node $c$}
Input: $Q$: blocker set, $c$: newly chosen blocker node, $S$: set of sources
\begin{algorithmic}[1]
\Statex \textbf{(only for $c$)}
\State \textbf{Local Step at $c$:} create $list_c$ to store the ID of each source $x \in S$ such that $score_x (c) \neq 0$; \textbf{for each} $x \in S$ \textbf{do} set $score_x (c) \la 0$; set $score(c) \la 0$	\label{algDes:send1}
\State \textbf{Send:} \textbf{Round $i$:} let $\langle x \rangle$ be the $i$-th entry in $list_c$; send $\langle x \rangle$ to $c$'s children in $T_x$.	\label{algDes:send2} \vspace{.05in}
\Statex \textbf{(round $r > 0:$ for vertices $v \in V - Q - \{c\}$)}
\State \textbf{send[lines~\ref{algDes:send3}-\ref{algDes:send4}]:}
\If{$v$ received a message $\langle x \rangle$ in round $r-1$}	\label{algDes:send3}
	\State \textbf{if} $v \neq x$ \textbf{then} send $\langle x \rangle$ to $v$'s children in $T_x$		\label{algDes:send4}
\EndIf
\State \textbf{receive[lines~\ref{algDes:receive1}-\ref{algDes:receive2}]:}
\If{$v$ receives a message $\langle x \rangle$}	\label{algDes:receive1}
	\State $score(v) \la score(v) - score_x (v)$; $score_x (v) \la 0$	\label{algDes:receive2}
\EndIf
\end{algorithmic}  \label{algDes}
\end{algorithm}

\vspace{-.15in}

Here, we perform no precomputation but instead  in Algorithm~\ref{algDes} we use the property in Lemma~\ref{lemma:cOutTree} to
develop a method similar to the one for updates at ancestors. 
Initially $c$ creates a list, $list_c$, where it adds the IDs of all the source nodes $x$ such that $c$ lies in tree $T_x$.
In round $i$, $c$ sends the $i$-th entry $\langle x \rangle$ in $list_c$ to all its children in $T_x$.
Since $T$ (in Lemma~\ref{lemma:cOutTree}) is a tree, every node $v$ receives at most one message in a given round $r$. 
If $v$ receives the message for source $x$ in round $r$, it forwards this message to all its children in $T_x$ in
the next round, $r+1$, and also set its score for source $x$ to $0$.
Similar to the algorithm for updating ancestors of $c$~\cite{ARKP18}, it is readily seen
 that every descendant of $c$ in every tree $T_x$ receives a message for 
$x$ by round $k+h-1$.

\begin{lemma}	\label{lemma:algDes}
Algorithm~\ref{algDes} correctly updates the scores of all nodes $v$ in every tree $T_x$ in which $v$ is a descendant
of $c$ in $k + h -1$ rounds.
\end{lemma}


\vspace{-.05in}

\section{Simplified Versions of Short-Range Algorithms in~\cite{HNS17} }\label{sec:sssp}

We describe here simplified versions of the {\it short-range}  and {\it short-range-extension} algorithms
 used in the randomized
$\tilde{O}(n^{5/4})$ round APSP algorithm in Huang et al.~\cite{HNS17}.
Our short-range Algorithm~\ref{algSSSP}  is 
inspired by
the pipelined APSP algorithm in~\cite{AR18}
and is much simpler 
than it since it is for a single source.

Given a hop-length $h$ and a source vertex $x$, the
 short-range algorithm in~\cite{HNS17} computes the $h$-hop shortest path distances from source $x$
 in a graph $G'$ (obtained through `scaling') where $\Delta \leq n-1$. The scaled graph has different edge weights for different sources, and hence $h$-hop APSP is computed
 through $n$ $h$-hop SSSP (or {\it short-range}) computations, each of which
runs with {\it dilation} (i.e., number of rounds)  $\tilde{O}(n\sqrt{h})$ and {\it congestion}
(i.e., maximum number of messages along an edge) $O(\sqrt h)$. By running this algorithm using each vertex as source, $h$-hop APSP is  computed in $G'$ in $O(n \sqrt h)$ rounds w.h.p. in $n$ using a result
in Ghaffari's 
framework~\cite{Ghaffari15}, which gives a randomized method to execute  this collection of different short-range executions simultaneously in 
$\tilde{O}(\text{dilation} + n \cdot \text{congestion}) = \tilde{O}(n \sqrt h)$ rounds.

The short-range algorithm in~\cite{HNS17} for a given source runs in two stages:.
Initially every zero edge-weight is increased to a positive value $\alpha = 1/\sqrt{h}$ and then $h$-hop SSSP is computed
using a BFS variant in $\tilde{O}(n/\alpha) = \tilde{O}(n\sqrt{h})$ rounds.
This gives an approximation to the $h$-hop SSSP
where the additive error is at most $h\alpha = \sqrt{h}$.
This error is then fixed by running Bellman-Ford algorithm~\cite{Bellman58} for $h$ rounds.
The total round complexity of this SSSP algorithm is $\tilde{O}(n\sqrt{h})$ and the congestion is $O(\sqrt{h})$.

\vspace{-.05in}

\begin{algorithm}[H]
\scriptsize
\caption{\scriptsize Round $r$ of short-range algorithm for source $x$ \\ (initially $d^* \la 0$; $l^* \la 0$ at source $x$)}
\begin{algorithmic}[1]
\Statex \textbf{(at each node $v \in V$)}
\State \textbf{send:} \textbf{if} $\lceil d^*\cdot \sqrt{h} + l^* \rceil = r$ \textbf{then} send $(d^*, l^*)$ to all the neighbors  \vspace{0.05in}	\label{algSSSP:send}
\State \textbf{receive [Steps~\ref{algSSSP:receiveStart}-\ref{algSSSP:receiveEnd}]:}  let $I$ be the set of incoming messages	\label{algSSSP:receiveStart}
\For{{\bf each} $M \in I$}
\State let $M = (d^{-}, l^{-})$ and let the sender be $y$.
\State $d \la d^{-} + w (y,v)$; $l \la l^{-} + 1$
\State {\bf if} $d < d^*$ or $(d = d^* \text{ and } l < l^*)$ \textbf{then} set $d^* \la d$; $l^* \la l$ \label{algSSSP:receiveEnd}
\EndFor
\end{algorithmic} \label{algSSSP}
\end{algorithm}

\vspace{-.1in}

We now describe our simplified short-range algorithm (Algorithm~\ref{algSSSP}) which has the same dilation $O(n\sqrt{h})$
and congestion $O(\sqrt h)$.
Here $d^*$ is the current best estimate for the shortest path distance from $x$ at node $v$ and $l^*$
is the hop-length of the corresponding path. 
Source node $x$ initializes $d^*$ and $l^*$ values to zero and sends these values to its
neighbors in round $0$ (Step~\ref{algSSSP:send}).
At the start of a round $r$, each node $v$ checks if its current  
$d^*$ and $l^*$ values satisfy 
$\lceil d^*\cdot \sqrt{h} + l^* \rceil = r$, and if so,  it
sends this estimate to each of its neighbors. 
To bound the number of such messages $v$ sends throughout the entire execution, we note that $v$ will  send another
message in a future round only if it receives a smaller $d^*$ value with higher $\lceil d^*\cdot \sqrt{h} + l^* \rceil$ value.
But since $l^* \leq h$ and $d^*$ values are non-negative integers,
$v$ can send at most $\sqrt{h}$ messages to its neighbors throughout the entire execution.

We now establish that vertex $v$ will receive the message that creates the pair $d^*, l^*$ at $v$ before round
$\lceil d^*\cdot \sqrt{h} + l^* \rceil$, and hence will be able to perform the send in 
Step~\ref{algSSSP:send} of  Algorithm~\ref{algSSSP}.

\begin{lemma}	\label{lemma:sendBound}
Let $\pi^*_{x,v}$ be a path from source $x$ to vertex $v$ with the minimum number of hops among all
$h$-hop shortest paths from $x$ to $v$. Let $\pi^*_{x,v}$ have $l^*$ hops and weight (distance) $d^*$.
If $v$ receives the message for the pair $d^*$, $l^*$ 
in round $r$ then  $r <  \lceil d^*\cdot \sqrt{h} + l^* \rceil$.
\end{lemma}

\begin{proof}
We show this by induction on round $r$. The base case is trivially satisfied since $x$ already knows $d^*$ and $l^*$ values at the start (Round $0$).

Assume inductively that the lemma holds at all vertices up to round $r-1$. 
Let $y$ be the predecessor of $v$ on the path $\pi^*_{x,v}$.
Then $y$ must have received the message for its pair $(d^* - w(y,v)$, $l^* - 1)$ in a  round
$r' < r$.
Let $k=\lceil (d^* - w(y,v))\cdot  \sqrt{h} + l^* - 1 \rceil$. Then,
 $r' < k$ by the inductive assumption. 
 So $y$ will send the message 
 $(d^* - w(y,v)$, $l^* - 1)$ to $v$ in
 round $r=k$ in Step~\ref{algSSSP:send} of
 Algorithm~\ref{algSSSP}. But
 $k=  \lceil (d^* - w(y,v))\cdot  \sqrt{h} + l^* -1 \rceil < \lceil (d^* - w(y,v))\cdot  \sqrt{h} + l^* \rceil\leq  \lceil d^*\cdot \sqrt{h} + l^* \rceil$, since $w(y,v) \geq 0$. Hence the round $r$ in which $v$ receives
 the message for the pair $d^*$, $l^*$   is less than $\lceil d^*\cdot \sqrt{h} + l^* \rceil$.
 
This establishes the induction step and the lemma.
\end{proof}

If shortest path distances are bounded by $\Delta$, Algorithm~\ref{algSSSP}  runs in $\lceil \Delta\cdot \sqrt{h} + h \rceil$ rounds with congestion at most $\sqrt h$.
And if $\Delta \leq n-1$ (as in~\cite{HNS17}),
then we can compute shortest path distances from $x$ to every node $v$ in $O(n\sqrt{h})$ rounds.

We can similarly simplify the short-range-extension algorithm in~\cite{HNS17},
 where some nodes already know their distance from source $x$ and the goal is to 
compute shortest paths from $x$ by extending these already computed shortest paths to $u$ by another $h$ hops.
To implement this, we only need to modify the initialization in Algorithm~\ref{algSSSP} so that  each such node $u$ 
initializes $d^*$ with  this already computed distance.
The round complexity is again $O(\Delta \sqrt{h})$ and the congestion per source is $O(\sqrt{h})$.
This gives us the following result.

\begin{lemma}
Let $G = (V,E)$ be a directed or undirected graph, where all edge weights are non-negative distances 
(and zero-weight edges are allowed),
and where shortest path distances are bounded by $\Delta$.
Then by using Algorithm~\ref{algSSSP}, we can compute $h$-hop SSSP and $h$-hop extension in $O(\Delta \sqrt{h})$ rounds with congestion bounded by $\sqrt{h}$.
 \end{lemma}
 
 As in~\cite{HNS17} we can now combine our Algorithm~\ref{algSSSP} with Ghaffari's randomized framework~\cite{Ghaffari15}
 to compute $h$-hop APSP and $h$-hop extensions (for all source nodes) in $\tilde{O}(\Delta\sqrt{h} + n\sqrt{h})$ rounds w.h.p. in $n$.
 The result can be readily  modified to include the number of sources, $k$, by sending the current estimates $(d^*, l^*)$ in round
 $\lceil d^*\cdot \gamma + l^*\rceil$ , where $\gamma =\sqrt{hk/\Delta}$  as in pipelined algorithm in~\cite{AR18}  (instead of  $\lceil d^*\cdot \sqrt{h} + l^* \rceil$), and 
 the resulting algorithm runs in 
 $O(\sqrt{\Delta hk})$ rounds with congestion bounded by $\sqrt{\Delta h/k}$. Then we can compute $h$-hop
 $k$-SSP and $h$-hop extensions for all $k$ sources in $\tilde{O}(\sqrt{\Delta hk})$ rounds.

\section{Additional Results}\label{sec:additional}

\subsection{A Faster Randomized Algorithm for Weighted APSP with Arbitrary Edge-Weights}\label{sec:randAPSP}

We adapt the randomized framework of Huang et al.~\cite{HNS17} to obtain a faster 
randomized algorithm for weighted APSP with arbitrary edge weights. Our randomized algorithm 
runs in $\tilde{O}(n^{4/3})$ rounds w.h.p. in $n$, improving on the previous best bound of
 $\tilde{O}(n^{3/2})$ rounds (which is deterministic) in Agarwal et al.~\cite{ARKP18}.
 We describe our randomized algorithm below.

As described in Section~\ref{sec:sssp}, Huang et al.\cite{HNS17} use two algorithms {\it short-range} and
{\it short-range-extension} for integer-weighted
APSP for which they have randomized algorithms that run in $\tilde{O}(n \sqrt h)$ rounds w.h.p. 
in $n$.  (We presented  simplified versions of these two algorithms in Section~\ref{sec:sssp}.)
Since we consider arbitrary edge weights here, we will instead use $h$ rounds of the Bellman-Ford algorithm~\cite{Bellman58} for both steps, which will take $O(kh)$ rounds  for $k$ source nodes.

We keep the remaining steps in~\cite{HNS17} unchanged:
These steps  involve having every `center' $c$ broadcast its estimated shortest distances, $\delta (c',c)$,
 from every other
center $c'$,  and each source node $x \in S$ sending its correct shortest distance, 
$\delta (x,c)$,
to each center $c$. (The set of {\it centers} is a random subset of vertices in $G$ of size $\tilde{O}(\sqrt n)$.)
These steps are shown  in~\cite{HNS17}  to take $\tilde{O}(n + \sqrt{nkq})$ rounds in total w.h.p. in $n$, where $q = \Theta(\frac{n\log n}{h})$.
This gives an  overall round complexity  $\tilde{O}(kh + n + \sqrt{nkq})$ for our algorithm.
Setting $h = n^{2/3}/k^{1/3}$ and $q = n^{1/3}k^{1/3}\log n$, we obtain the desired bound of $\tilde{O}(n + n^{2/3}k^{2/3})$ in Theorem~\ref{thm:arbitrary}.

\subsection{An $\tilde{O}(n)$-Rounds $(1+\epsilon)$ Approximation Algorithm for Weighted APSP with Non-negative Integer Edge-Weights}	\label{sec:approx}

Here we deal with the problem of finding $(1+\epsilon)$-approximate solution to the weighted APSP problem.
If edge-weights are strictly positive, the following result is known.

\begin{theorem}[~\cite{Nanongkai14, LP15}]	\label{thm:result}
There is a deterministic algorithm that computes $(1+\epsilon)$-approximate APSP on graphs with positive  polynomially bounded integer edge weights in $O((n/\epsilon^2)\cdot \log n)$ rounds.
\end{theorem}

The above result  does not hold when \textit{zero weight edges} are present.
Here we match the deterministic $O((n/\epsilon^2)\cdot \log n)$-round bound for this problem with
an algorithm that also handles zero edge-weights.

We first compute reachability between all pairs of vertices connected by zero-weight paths.
This is readily computed in $O(n)$ rounds, e.g., using~\cite{LP13,PR18}
while only considering only the zero weight edges 
(and ignoring the other edges).

We then  consider shortest path distances between pairs of
 vertices that have no zero-weight path connecting them.
The weight of any such path is at least $1$. To approximate these paths we
 increase the zero edge-weights to $1$ and transform every non-zero edge weight $w(e)$ to $n^2\cdot w(e)$.
Let this modified graph be $G' = (V, E, w')$ .
Thus the weight of an $l$-hop path $p$ in $G'$, $w'(p)$, satisfies $w'(p) \leq w(p) \cdot n^2 + l$.
Since the modified graph $G'$ has polynomially bounded positive edge weights, we can use
 the result in Theorem~\ref{thm:result}  to
compute $(1+\epsilon/3)$-approximate APSP on this graph in $\tilde{O}(9n/\epsilon^2)$ rounds.

 Fix a pair of vertices $u, v$. Let $p$ be a shortest path from $u$ to $v$ in $G$, and let its
 hop-length be $l$.
 Then $w' (p) \leq n^2 \cdot w(p) + l$.
 Let $p'$ be a $(1+\epsilon/3)$-approximate shortest path from $u$ to $v$, and let its hop-length be $l$.
 Then $w' (p') \leq (1+\epsilon/3)\cdot w'(p) \leq (1+\epsilon/3)\cdot (n^2\cdot w(p) + l)$.
 Dividing $w'(p')$ by $n^2$ gives us $w'(p')/n^2 < w(p) (1+\epsilon/3) +  (l/n^2) (1+\epsilon/3) < w(p) + w(p)\epsilon/3 + 2/n \leq w(p) (1+\epsilon/3) + 2\epsilon/3 \leq w(p) (1+\epsilon)$ (as long as $\epsilon > 3/n$ and since $w(p) \geq 1$), and this 
establishes Theorem~\ref{thm:algApprox}. 
 
\section{Conclusion}\label{sec:conclusion}

We have presented  new improved deterministic distributed algorithms for
weighted shortest paths (both APSP, and for $k$ sources)
in graphs with moderate non-negative integer weights. These results build on our recent pipelined algorithm for weighted shortest paths~\cite{AR18}.
We have also presented simplications to two procedures in the randomized APSP algorithm in Huang et al.~\cite{HNS17} by streamlining the pipelined APSP
algorithm in~\cite{AR18} for SSSP versions.
 A key feature of our shortest path algorithms is that they can handle zero-weighted edges, which are known 
to present a challenge in the design of distributed algorithms for non-negative integer weights. We have
also presented a faster and more streamlined algorithm to compute a blocker set,  an improved randomized APSP (and $k$-SSP) algorithm for arbitrary edge-weights, and
an approximate APSP algorithm that can handle zero-weighted edges.

 An important area for further research is to investigate
further improvements to the deterministic distributed computation of a blocker set, beyond the algorithm in~\cite{ARKP18} and the  improvements we have presented here.

\bibliographystyle{abbrv}
\bibliography{references}

\begin{thebibliography}{10}

\bibitem{AR18}
U.~Agarwal and V.~Ramachandran.
\newblock A faster deterministic distributed algorithm for weighted apsp
  through pipelining.
\newblock {\em arXiv preprint arXiv:1807.08824.v2}, 2018.

\bibitem{ARKP18}
U.~Agarwal, V.~Ramachandran, V.~King, and M.~Pontecorvi.
\newblock A deterministic distributed algorithm for exact weighted all-pairs
  shortest paths in $\tilde{O} (n^{3/2})$ rounds.
\newblock In {\em Proc. PODC}, pages 199--205. ACM, 2018.

\bibitem{Bellman58}
R.~Bellman.
\newblock On a routing problem.
\newblock {\em Quarterly of applied mathematics}, 16(1):87--90, 1958.

\bibitem{Dijkstra59}
E.~W. Dijkstra.
\newblock A note on two problems in connexion with graphs.
\newblock {\em Numerische mathematik}, 1(1):269--271, 1959.

\bibitem{Elkin17}
M.~Elkin.
\newblock Distributed exact shortest paths in sublinear time.
\newblock In {\em Proc. STOC}, pages 757--770. ACM, 2017.

\bibitem{Ghaffari15}
M.~Ghaffari.
\newblock Near-optimal scheduling of distributed algorithms.
\newblock In {\em Proc. PODC}, pages 3--12. ACM, 2015.

\bibitem{GL18}
M.~Ghaffari and J.~Li.
\newblock Improved distributed algorithms for exact shortest paths.
\newblock In {\em Proc. STOC}, pages 431--444. ACM, 2018.

\bibitem{HNS17}
C.-C. Huang, D.~Nanongkai, and T.~Saranurak.
\newblock Distributed exact weighted all-pairs shortest paths in $\tilde{O}
  (n^{5/4})$ rounds.
\newblock In {\em Proc. FOCS}, pages 168--179. IEEE, 2017.

\bibitem{King99}
V.~King.
\newblock Fully dynamic algorithms for maintaining all-pairs shortest paths and
  transitive closure in digraphs.
\newblock In {\em Proc. FOCS}, pages 81--89. IEEE, 1999.

\bibitem{KN18}
S.~Krinninger and D.~Nanongkai.
\newblock A faster distributed single-source shortest paths algorithm.
\newblock In {\em Proc. FOCS}. IEEE, 2018.

\bibitem{LP15}
C.~Lenzen and B.~Patt-Shamir.
\newblock Fast partial distance estimation and applications.
\newblock In {\em Proc. PODC}, pages 153--162. ACM, 2015.

\bibitem{LP13}
C.~Lenzen and D.~Peleg.
\newblock Efficient distributed source detection with limited bandwidth.
\newblock In {\em Proc. PODC}, pages 375--382. ACM, 2013.

\bibitem{Nanongkai14}
D.~Nanongkai.
\newblock Distributed approximation algorithms for weighted shortest paths.
\newblock In {\em Proc. STOC}, pages 565--573. ACM, 2014.

\bibitem{PR18}
M.~Pontecorvi and V.~Ramachandran.
\newblock Distributed algorithms for directed betweenness centrality and all
  pairs shortest paths.
\newblock {\em arXiv preprint arXiv:1805.08124}, 2018.

\end{thebibliography}

\end{document}